\documentclass[10pt,doublecolumn]{IEEEtran}

\pdfoutput=1

\usepackage{amsthm,amsfonts}
\usepackage{amsmath}
\usepackage{amssymb}
\usepackage{graphicx}
\usepackage{makeidx}

\newtheorem{mydef}{Definition}
\newtheorem{myteo}{Theorem}
\newtheorem{mypro}{Proposition}
\newtheorem{mycor}{Corollary}
\newtheorem{mylem}{Lemma}
\newtheorem{myexe}{Example}

\begin{document}

\title{The Packing Radius of a Code and Partitioning  Problems: the Case for Poset Metrics}

\author{Rafael Gregorio Lucas D'Oliveira, Marcelo Firer
\thanks{R. G. L. D'Oliveira was partially supported by CNPq. M. Firer was partially supported by FAPESP grant 2007/ 56052-8. }%
\thanks{R. G. L. D'Oliveira is with IMECC-UNICAMP, Universidade Estadual de Campinas, CEP 13083-
859, Campinas, SP, Brazil (e-mail: \newline rgldoliveira@gmail.com) }
\thanks{M. Firer is with IMECC-UNICAMP, Universidade Estadual de Campinas, CEP 13083-
859, Campinas, SP, Brazil (e-mail: mfirer@gmail.com)}}

\date{\today}
\maketitle

\begin{abstract}
Until this work, the packing radius of a poset code was only known in the cases where the poset was a chain, a hierarchy, a union of disjoint chains of the same size, and for some families of codes. Our objective is to approach the general case of any poset. To do this, we will divide the problem into two parts.

The first part consists in finding the packing radius of a single vector. We will show that this is equivalent to a generalization of a famous NP-hard problem known as ``the partition problem''. Then, we will review the main results known about this problem giving special attention to the algorithms to solve it. The main ingredient to these algorithms is what is known as the differentiating method, and therefore, we will extend it to the general case.

The second part consists in finding the vector that determines the packing radius of the code. For this, we will show how it is sometimes possible to compare the packing radius of two vectors without calculating them explicitly.
\end{abstract}

\begin{IEEEkeywords}
Error correction codes, poset codes.
\end{IEEEkeywords}

\section{Introduction}

An important concept of coding theory is that of the packing radius of a code. When using the Hamming metric this concept is overshadowed by that of the minimum distance since it is determined completely by it, i.e. if $C$ is a code and $d_H(C)$ is the minimum distance of $C$ in the Hamming metric, it is well known that the packing radius of $C$ is
 \[ R_{d_H}(C) = \left \lfloor \frac{d_H (C) - 1}{2} \right \rfloor .\]

In this work we will consider the problem of finding the packing radius in the case of  poset metrics. These metrics where first introduced by Brualdi et al. \cite{brua95} generalizing on the work of Neiderreiter \cite{nied91}. An interesting property of these metrics is that the packing radius is not necessarily determined by the minimum distance. What we do have, nonetheless, is the following inequality
\[ \left \lfloor \dfrac{d_P (C)-1}{2} \right \rfloor \leq R_{d_P} (C) \leq d_P (C) -1 .\]
It is known that the upper bound is attained when the poset is a chain \cite{fire10}.

Until this work, to the authors' knowledge, the packing radius of a poset code was only known in the following cases: chain posets \cite{fire10}, hierarchical posets \cite{feli11}, disjoint union of chains of the same size \cite{pane07}, and for some families of codes \cite{fire11}. We will approach the general poset case. To do this we will divide our problem in two.

The first part consists in determining the packing radius of a single vector. We will see that this is equivalent to solving a generalization, which we will call ``the poset partition problem'', of a famous NP-hard problem known as ``the partition problem''. We will then take a look at the best known algorithm for solving the partition problem and generalize it to the poset partition problem. The first time the problem of finding the packing radius of a poset code was identified, in some sense, as a partitioning problem was in \cite{yoon04}.

The second part consists in finding which code-word determines the packing radius of the code. To do this we will show how some times it is possible to compare the packing radius of two vectors without calculating them explicitly.

\section{Preliminaries}

\subsection{The Poset Metric}

Let $[n] = \{1,2, \ldots , n \}$ be a finite set and $ \preceq$ be a partial order on $[n]$. We call the pair $P = ( [n] , \preceq)$ a poset and often identify $P$ with $[n]$. An ideal in $P$ is a subset $J \subseteq P$ with the property that if $x \in J$ and $y \preceq x$ then $y \in J$. The ideal generated by a subset $X \subseteq P$ is the smallest ideal containing $X$ and is denoted by $\left \langle X \right \rangle$. A poset is called a chain if every two elements are comparable, and an anti-chain if none are. The length of an element $x \in P$ is the cardinality of the largest chain contained in $\left \langle \{x \} \right \rangle$.

Let $q$ be the power of a prime, $\mathbb{F}_q$ the field with $q$ elements and ${\mathbb{F}_q^n}$ the vector space of $n$-tuples over $\mathbb{F}_q$. We denote the coordinates of a vector $x \in {\mathbb{F}_q^n}$ by $x = ( x_1, x_2, \ldots , x_n )$. 

A poset $P = ( [n] , \preceq)$ induces a metric $d_P$, called the $P$-distance, in ${\mathbb{F}_q^n}$ defined as 
\[ d_P (v,w) = \left | \left \langle  supp(v-w)  \right \rangle \right | \]
where $supp(x) = \{ i \in [n] : x_i \neq 0\}$. The distance $\omega_P(v) = d_P(v,0)$ is called the $P$-weight of $v$. 

Note that if $P$ is an anti-chain then $d_P$ is the Hamming distance. Because of this, when $P$ is an anti-chain we will denote it by $H$.

Given a linear code $C \subseteq {\mathbb{F}_q^n}$ and a poset $P = ( [n] , \preceq)$, we define the minimum distance of $C$ as $d_P (C) = min \{ \omega_P (v) : v \in C - \{0 \} \}$. The packing radius of $C$ is denoted by $R_{d_P} (C)$ and is defined as the largest positive integer such that \[B_P (x,R_{d_P} (C)) \cap B_P (y,R_{d_P} (C)) = \varnothing \] for every $x,y \in C$. Since $C$ is linear, $z=x-y \in C$ and therefore the packing radius is the largest positive integer such that \[B_P (0,R_{d_P} (C)) \cap B_P (z,R_{d_P} (C)) = \varnothing \] for every $z \in C$. 

It is well known that in the Hamming case, the packing radius of a linear code $C \subseteq {\mathbb{F}_q^n}$ is given by   \[ R_{d_H}(C) = \left \lfloor \frac{d_H (C) - 1}{2} \right \rfloor .\]
In the general case, however the following inequality is true:
\[ \left \lfloor \dfrac{d_P (C)-1}{2} \right \rfloor \leq R_{d_P} (C) \leq d_P (C) -1 .\]

\section{The Packing Radius of a Vector}

We begin this section by defining the packing radius of a vector.

\begin{mydef}
Let $x \in \mathbb{F}_q^n$ and $d$ be a metric over $\mathbb{F}_q^n$. The \textbf{packing radius of $x$} is the largest integer $r$ such that 
\[ B(0,r) \cap B(x,r) = \emptyset \]
and is denoted by $R_d (x)$.
\end{mydef}

Note that the packing radius of $x$ is the packing radius of the code $C = \{ 0,x \}$.

Next, we show that the packing radius of a linear code is the smallest of the packing radius of its code-words.

\begin{mypro}
Let $C \subseteq \mathbb{F}_q^n$ be a linear code and $d$ a metric over $\mathbb{F}_q^n$. Then, 
\[ R_d (C) = \min_{x\in C^*} R_d (x) .\]
\end{mypro}

\begin{proof}
Let $y \in C^*$ such that $R_d (y) = \underset{x\in C^*}{\operatorname{min}} R_d (x)$. By definition, for every $x \in C^*$
\[ B(0, R_d (x) ) \cap B(x, R_d (x) ) = \emptyset .\]
But, $R_d (y) \leq R_d(x)$, and therefore, 
\[ B(0, R_d (y) ) \cap B(x, R_d (y) ) = \emptyset .\]
Since
\[ B(0, R_d (y)+1 ) \cap B(y, R_d (y)+1 ) = \emptyset ,\]
$R_d (C) = R_d (y)$.
\end{proof}

This result motivates the following definition:

\begin{mydef}
Let $C \subseteq \mathbb{F}_q^n$ be a linear code and $d$ a metric over $\mathbb{F}_q^n$. A code-word $x \in C^*$ such that $R_d (C) = R_d (x)$ is called a \textbf{packing vector} of $C$.
\end{mydef}

The packing vector of a linear code determines its packing radius. In this section, we will focus on the problem of finding the packing radius of a vector, and leave the problem of finding the packing vector of a linear code, and thus its packing radius, for later.

A possible brute force algorithm for finding the packing radius of a vector $v$ is the following: Start with $r=1$. List the elements of $B(0, r)$ and $B(v , r)$. If there are any repeated elements terminate, otherwise add one to $r$ and start again. After the algorithm is terminated the packing radius will be $r-1$.

\begin{myexe}
Let $P$ be a chain of height three such that $1 \preceq 2 \preceq 3$ and $d_P$ the metric induced by this poset on $\mathbb{F}_2^3$. Lets find the packing radius $v = 001$ using the brute force algorithm described above. The following table lists the elements found in the spheres of size $r$ and center $c$.
\[
\begin{tabular}{ | c | c | c | } \hline
& c=000 & c=001 \\ \hline
r=1 & 100 & 101 \\ \hline
r=2 & 110 010 & 111 011 \\ \hline
r=3 & 111 101 001 011 & 110 100 000 010 \\ \hline
\end{tabular}
\]
Since there are no repetitions of elements on the two columns until $r=3$, the packing radius of $v$ is $R_{d_P} (001) = 2$.
\end{myexe}

In the example given above, the vectors of the two columns are related by a translation $T$ of the form $T(x) = v-x$ with the property that $x \in S(0,r)$ iff $T(x) \in S(v,r)$. The reason for this is that the poset metric above, as all other poset metrics, is translation invariant, i.e. for all $ x,y,z \in \mathbb{F}_q^n$, $d_P (x,y) = d_P (x+z , y+z )$.

We make the following definition:

\begin{mydef}
Let $x, v \in \mathbb{F}_q^n$. We define the \textbf{complement of $x$ with respect to $v$} as $x^v = v-x$.
\end{mydef}

We now make our last statement precise.

\begin{mylem}
Let $d$ be a translation invariant metric over $\mathbb{F}_q^n$, $x,v \in \mathbb{F}_q^n$, and $r \in \mathbb{R}$, such that $r \geq 0$. Then $x \in S(0,r)$ iff $x^v \in S(v,r)$.
\end{mylem}

\begin{proof}
Since $d$ is translation invariant, 
\[ d(v,x^v) = d(v-x^v,0) = d(x, 0) .\]
\end{proof}

Therefore, when the metric being considered is translation invariant, as is the case for a poset metric, the packing radius of a vector $v \in \mathbb{F}_q^n$ can be found in the following way: We make two lists; in the first we put the vectors of $\mathbb{F}_q^n$ ordered by weight, and for every vector of this list we put its complement in the second list. When a vector $x$ which has already appeared in the second list appears also in the first, we can determine the packing radius of $v$ which will be $\omega_P (x) - 1$.

\begin{myexe}
Using the method just described on the last example we have:
\[
\begin{tabular}{ | c | c | } \hline
List 1 & List 2 \\ \hline
100 & 101 \\ \hline
110 & 111 \\ \hline
010 & 011 \\ \hline
111 & 110 \\ \hline
\end{tabular}
\]
Since $111$ already appeared in the second row, $R_{d_P}(001) = \omega_P (111) - 1 = 2$.
\end{myexe}

We will use the following notation to simplify our expressions.

\begin{mydef}
Let $\omega$ be a weight over $\mathbb{F}_q^n$ and $x,y \in \mathbb{F}_q^n$. We denote the maximum weight between $x$ and $y$ as
\[ \omega^\vee (x,y) = \max \{ \omega (x) , \omega (y) \} .\]
\end{mydef}

In general we have the following expression for the packing radius:

\begin{myteo}
Let $d$ be a translation invariant metric over $\mathbb{F}_q^n$ and $v \in \mathbb{F}_q^n$. Then,
\[ R_d (v) = \min_{x \in \mathbb{F}_q^n} \{ \omega^\vee (x,x^v) \} - 1 .\]
\end{myteo}

\begin{proof}
Let $y \in \mathbb{F}_q^n$ be such that
\[ \omega^\vee (y,y^v) = 	 \min_{x \in \mathbb{F}_q^n} \{ \omega^\vee (x,x^v) \} \]
and $R = \omega^\vee (y,y^v)$. Suppose there exists $ z \in B(0, R-1) \cap B(v, R-1)$. Since $z \in B(0, R-1)$, we have $\omega (z) \leq R-1$, and therefore, 
\[ \omega (z) < \omega^\vee (y,y^v) .\]
Since $z \in B(v, R-1)$, by Lemma 1, $z^v \in B(0, R-1)$. But then, $\omega (z^v) \leq R-1$, and therefore,
\[ \omega (z^v) < \omega^\vee (y,y^v) .\]
But if this is true, then 
\[ \omega^\vee (z,z^v) < \omega^\vee (y,y^v) ,\]
a contradiction. Thus, $B(0, R-1) \cap B(v, R-1) = \emptyset$, and since $y \in B(0,R) \cap B(v,R)$, the theorem is proved.
\end{proof}

The following definition is quite natural.

\begin{mydef}
Let $d$ be a translation invariant metric over $\mathbb{F}_q^n$ and $v \in \mathbb{F}_q^n$. We say $y \in \mathbb{F}_q^n$ is a \textbf{radius vector} of $v$ if 
\[ \omega^\vee (y,y^v) = 	 \min_{x \in \mathbb{F}_q^n} \{ \omega^\vee (x,x^v) \} .\]
\end{mydef}

We can therefore find the packing radius of a vector $v \in \mathbb{F}_q^n$ by finding a radius vector of $v$. If we search the whole of $\mathbb{F}_q^n$ our search space will have size $q^n$. Our following results will reduce the size of our search space. We will be using the following notation.

\begin{mydef}
The \textbf{Iverson bracket}, introduced in \cite{iver62}, is defined as:
\[ [P] = \left\{\begin{matrix}
1 & \text{if P is true}\\ 
0 & \text{if P is false}
\end{matrix}\right.
\]
where $P$ is a statement that can be true or false.
\end{mydef}

This notation was popularized by Knuth in \cite{knut92}.

\begin{mylem}
Let $P$ be a poset and $v \in \mathbb{F}_q^n$. Then, there exists $x \in \mathbb{F}_q^n$ such that:
\begin{enumerate}
\item $x$ is a radius vector of $v$ ;
\item $supp(x) \subseteq supp(v)$ ;
\item for all $i \in P$, $x_i = v_i$ or $x_i = 0$.
\end{enumerate}
\end{mylem}

\begin{proof}
Let $z \in \mathbb{F}_q^n$ be a radius vector of $v$. Define $x \in \mathbb{F}_q^n$ such that
\[ x_i = v_i [ z_i \neq 0 ] .\]
By its definition, $x$ satisfies the third condition. To show that it satisfies the second condition, let $i \in supp(x)$. Then, $x_i \neq 0$, but for this to happen we must have $v_i \neq 0$ and $z_i \neq 0$. Thus,
\[ i \in supp(v) \cap supp(z) .\]
Finally, we must show that it satisfies the first condition. Note that
\begin{align}
x_i^v &= v_i - x_i \nonumber \\
&= v_i - v_i [ z_i \neq 0 ] \nonumber \\
&= v_i (1-[z_i \neq 0]) \nonumber \\
&= v_i [z_i = 0] \nonumber .
\end{align}
Since $z_i^v = v_i - z_i$, 
\[ [z_i = 0 ] = [z_i^v = v_i ] \]
and therefore,
\[ x_i^v = v_i [z_i^v = v_i] .\]
Thus, if $i \in supp(x^v)$, then $x_i^v \neq 0$, and therefore, $v_i \neq 0$ and $z_i^v = v_i$, hence, $z_i^v \neq 0$. Therefore,
\[ supp(x^v) \subseteq supp(z^v) .\]
But we have already seen that
\[ supp(x) \subseteq supp(z) .\]
Thus,
\[ \omega_P (x) \leq \omega_P (z) \]
and
\[ \omega_P (x^v) \leq \omega_P (z^v) .\]
But then,
\[ \omega^\vee_P (x,x^v) \leq \omega^\vee_P (z,z^v) ,\]
from which the first condition follows.
\end{proof}

This lemma reduces our search space for a vector radius of $v \in \mathbb{F}_q^n$ to $2^{| supp(v) |}$ elements. The third condition of the lemma shows that the cardinality of the field is irrelevant to our problem, that is, the case $q=2$ is equivalent to the general case. The second condition shows that the structure of the poset outside of the ideal generated by the support of $v$ is also irrelevant. Thus, the packing radius of a vector is a property of the ideal generated by its support.

With this in mind, the following definitions are quite natural.

\begin{mydef}
Let $P$ be a poset and $A \subseteq P$. We define the \textbf{$P$-weight} of $A$ as 
\[ \omega_P (A) = | \langle A \rangle | ,\]
the cardinality of the Ideal generated by $A$.
\end{mydef}

\begin{mydef}
Let $P$ be a poset, $v \in \mathbb{F}_q^n$, and $A \subseteq supp(v)$. We define the \textbf{complement of $A$ with respect to $v$} as
\[ A^v = supp(v) - A .\]
\end{mydef}

Note that $(A, A^v)$ is a partition of $supp(v)$.

We are now ready to show that the problem of determining the radius vector of a vector is essentially a partitioning problem.

\begin{mylem}
Let $P$ be a poset and $v \in \mathbb{F}_q^n$. Then, 
\[ \min_{x \in \mathbb{F}_q^n} \{ \omega^\vee_P (x,x^v) \} = \min_{A \subseteq supp(v)} \{ \omega^\vee_P (A,A^v) \} \]
\end{mylem}

\begin{proof}
Define 
\[f_v : 2^{supp(v)} \mapsto \mathbb{F}_q^n\]
as the function that takes $A \subseteq supp(v)$ as input and outputs $f_v (A) = x$ such that
\[ x_i = v_i [ i \in A ] .\]
This function is injective and weight preserving, i.e. $\omega_P (A) = \omega_P (f_v (A))$. It also respects complementation, as, since $(A, A^v)$ is a partition of $supp(v)$, 
\[ f_v (A) + f_v (A^v) = v_i [ i \in A ] + v_i [i \in A^v] = v_i ,\]
and therefore,
\[ f_v (A^v) = f_v (A)^v . \]
Thus,
\[ \min_{A \subseteq supp(v)} \{ \omega^\vee_P (A,A^v) \} = \min_{x \in f_v (2^{supp(v)})	} \{ \omega^\vee_P (x,x^v) \} .\]
By Lemma 2, there exists a radius vector of $v$ in $f_v ( 2^{supp(v)} )$, and the result follows.
\end{proof}

We define a radius set analogously to the concept of a radius vector.

\begin{mydef}
Let $P$ be a poset and $v \in \mathbb{F}_q^n$. We say $X \subseteq supp(v)$ is a \textbf{radius set} of $v$ if
\[ \omega^\vee_P (X,X^v)  = \min_{A \subseteq supp(v)} \{ \omega^\vee_P (A,A^v) \} .\]
\end{mydef}

We will now reduce our search space even more. For this we will need the following definition:

\begin{mydef}
Let $P$ be a poset and $A \subseteq P$. We denote the set of maximal elements of $A$ by $M_A$.
\end{mydef}

The weight of a set is determined by its maximal elements, i.e. $\omega_P (A) = \omega_P(M_A)$ since the ideal generated by both are the same. It is quite clear then that the maximal elements must be important in determining the radius set.

\begin{mylem}
Let $P$ be a poset and $v \in \mathbb{F}_q^n$. Then, there exists a radius set $X$ of $v$ such that $(M_X , M_{X^v})$ is a partition of $M_{supp(v)}$.
\end{mylem}

\begin{proof}
Let $Z$ be a radius set of $v$. Define
\[ X = \langle M_Z \cap M_{supp(v)} \rangle \cap supp(v) .\]
First, we show that $X$ is a radius set of $v$. By definition,
\[ M_X = M_Z \cap M_{supp(v)} ,\]
and therefore,
\[ \langle X \rangle = \langle M_Z \cap M_{supp(v)} \rangle \subseteq \langle M_Z \rangle .\]
Thus,
\[ \omega_P (X) \leq \omega_P (Z) .\]

Now, let $i \in X^v$. By definition, $i \in supp(v)$ and $i \notin X$, hence,
\[ i \notin \langle M_Z \cap M_{supp(v)} \rangle .\]
Since $i \in supp(v)$, there exists $j \in M_{supp(v)}$ such that $i \leq j$. If $j \in M_Z$, then
\[ i \in \langle M_Z \cap M_{supp(v)} \rangle ,\]
a contradiction. Thus, $j \in M_{Z^v}$, and therefore, $i \in \langle M_{Z^v}$. Hence, $X^v \subseteq \langle M_{Z^v} \rangle$ and consequently
\[ \omega_P (X^v) \leq \omega_P (Z^v) .\]

Using both inequalities we have proven that $X$ is a radius set of $v$. Now, we need to prove that $(M_X , M_{X^v})$ is a partition of $M_{supp(v)}$. We do this by proving that
\[ M_{X^v} = M_{supp(v)} - M_X .\]

Let $i \in M_{X^v}$. Then, certainly, $i \notin M_X$. Suppose that $i \notin M_{supp(v)}$, then there exists $j \in M_{supp(v)}$ such that $i < j$. Since $j$ cannot be in $X^v$ we must have $j \in X$, but that would imply $i \in M_X$, a contradiction. Thus, $i \in M_{supp(v)}$, and therefore,
\[ M_{X^v} \subseteq M_{supp(v)} - M_X .\]

Now, let $i \in M_{supp(v)} - M_X$. Then, $i \in M_{supp(v)}$ and $i \notin M_X$. Since $i \notin X$ it is certainly in $X^v$, but it is also in $M_{supp(v)}$ and is therefore in $M_{X^v}$. Hence,
\[ M_{X^v} \supseteq M_{supp(v)} - M_X ,\]
and the lemma is proven.
\end{proof}

Thus, the problem of finding the packing radius of a vector has been transformed into a partitioning problem over the set of maximal elements of its ideals.

\begin{myteo}
Let $P$ be a poset and $v \in \mathbb{F}_q^n$. Then,
\[ R_{d_P} (v) = \min_{A,B \subseteq M_{supp(v)}} \{ max \{ \omega_P (A) , \omega_P (B) \} \} - 1 ,\]
where $(A,B)$ is a partition of $M_{supp(v)}$.
\end{myteo}

\begin{proof}
Apply Lemma 3 to Theorem 1, and then using the definition of radius set, apply Lemma 4.
\end{proof}

The size of our search spaced has been reduced to $2^{ | M_{supp(v)} | }$ elements.

We can make some changes on the notation in order to emphasize on the partitioning problem.

\begin{mydef}
Let $P$ be a poset and $M_P$ be the set of its maximal elements. We define the \textbf{packing radius of the poset} $P$ as
\[ R(P) = \min_{A,B \subseteq M_P} \{ max \{ | \langle A \rangle | ) , | \langle B \rangle | \} \} - 1 ,\]
where $(A,B)$ is a partition of $M_P$.
\end{mydef}

A partition that minimizes the expression above is called an \textbf{optimum partition}.

As a direct corollary to Theorem 2 we have:

\begin{mycor}
Let $P$ be a poset and $v \in \mathbb{F}_q^n$. Then,
\[ R_{d_p} (v) = R(supp(v)) .\]
\end{mycor}

The problem of finding the packing radius of a vector is then equivalent to the problem of finding the packing radius of a poset, which we will call the \textbf{poset partition problem}.  This problem is a generalization of the famous NP-hard problem known as ``the partition problem'', that will be introduced in the next section.

\section{The Partition Problem}

The \textbf{partition problem}, which we will also call the \textbf{classical partition problem}, is defined as follows: Given a finite list $S$ of positive integers, find a partition $(S_1 , S_2)$ of $S$ that minimizes
\[ \max \left\{ \sum_{x \in S_1} x , \sum_{y \in S_2} y \right\}. \]
This is equivalent to minimizing the \textbf{discrepancy}
\[ \Delta (S_1 , S_2) = \left| \sum_{x \in S_1} x - \sum_{y \in S_2} y \right| .\]

In principle, the lowest possible value for the discrepancy is $0$ or $1$, depending on the parity of sum of the elements of $S$. A partition that attains this value is called a \textbf{perfect partition}. The possible existence of these perfect partitions is very important for the algorithms used to solve the problem, for if one is found it determines the answer completely.

This problem is of great importance both from the practical and theoretical point of view. In \cite{karp72}, Karp proves that it is NP-hard. This being the case, unless P=NP, there does not exist an algorithm that solves all instances of the problem in polynomial time.

\subsection{The Karmarkar-Karp Heuristic}

The best heuristic known for this problem is the  \textbf{Karmarkar-Karp(KK)} Heuristic, also known as the \textbf{differencing method}, first introduced in \cite{karm82}. The method involves using the \textbf{differencing operation}: select two elements $x_i$ and $x_j$ from the list being partitioned and replace them by the element $| x_i - x_j |$. Doing this is equivalent to making the decision that they will go into different subsets. After applying this operation $n-1$ times a partition will have been made and its discrepancy will be the value of the single element left on the list.

Depending on which criterion is used to choose the two elements in each step to apply the differencing operation, many partitions can be obtained. For the classical partition problem, the best criterion known is the \textbf{largest differencing method(LDM)}, which chooses the two biggest elements.

\begin{myexe}
Let $(8,7,6,5,4)$ be the list being partitioned. The KK heuristic using the LDM criterion will have the following instances: $(8,7,6,5,4)$, $(6,5,4,1)$, $(4,1,1)$, $(3,1)$, $(2)$, giving a discrepancy of $2$ pertaining to the partition $(8,6),(7,5,4)$. In this case the heuristic does not find the optimal partition $(8,7),(6,5,4)$.
\end{myexe}

\subsection{Complete Karmarkar-Karp}

In \cite{korf98}, Korf shows how to extend an heuristic using the differencing operation into a \textbf{complete anytime algorithm}, i.e. an algorithm that finds better and better solutions the longer it runs, until it finally finds the optimal one. This algorithm is known as the \textbf{Complete Karmarkar-Karp(CKK)} algorithm.

At each instance the KK heuristic commits on placing two elements in different subsets, by replacing them with their difference. The only other option would be to commit to place them in the same set, replacing them by their sum. This results in a binary tree where each node represents a possible instance. A left branch on a node will lead to a node where the two chosen elements were replaced with their difference and a right branch to a node where they were replaced with their sum. If the list has $n$ elements, the whole tree will have $2^{n-1}$ terminal nodes corresponding to all the possible partitions.

\begin{figure}[htb]
\centering
\includegraphics[scale=0.3]{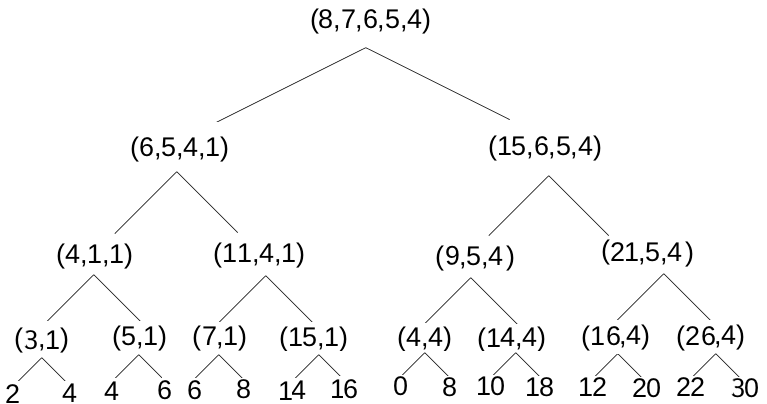}
\caption{Binary tree for partitioning $(8,7,6,5,4)$ using LDM.}
\end{figure}

If no perfect partition exists, the whole tree must be searched. In this (worst case) scenario, the running time of CKK is $O(2^n)$. If, otherwise, there is a perfect partition, then upon finding it, the search can stop. This being the case the order in which the terminal nodes are searched is important and the best results come from using the LDM criterion and giving preference to left branches. For a discussion on good ways to search the tree, see \cite{korf96}.

Another important aspect of searching the binary tree is pruning, i.e. sometimes it is not necessary to look at all the terminal nodes. For example, if the largest element of an instance is bigger than the sum of the remaining elements, then the best partition for that instance is to put the largest element in a subset and all the rest in the other one. Other criteria for pruning are:
\begin{itemize}
\item In a list with three elements the optimal partition is to put the biggest one in a subset and the other two in another.
\item In a list with four elements the KK heuristic finds the optimal partition.
\item In a list with five elements, if the KK heuristic does not find the optimal partition then the optimal partition is to put the two biggest elements in the same subset and the other ones in another.
\end{itemize}

\begin{figure}[htb]
\centering
\includegraphics[scale=0.4]{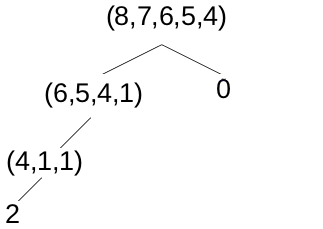}
\caption{Pruned tree for partitioning $(8,7,6,5,4)$ using LDM.}
\end{figure}

\section{The Packing Radius of a Poset}

In this section we will discuss how to determine the packing radius of a poset. First we will show that when the ideal generated by the maximal elements are pairwise disjoint, which we will call the \textbf{disjoint ideals case}, the poset partition problem is equivalent to the classic partition problem. Next we will approach the more general case of any poset, and after this we will generalize the differencing operation so that we can use it in the poset partition problem.

\subsection{The Disjoint Ideals Case}

We begin by showing that when the ideals generated by the maximal elements of a poset are pairwise disjoint, the poset partition problem is equivalent to the classic partition problem.

Note that this condition is equivalent to stating that each connected component of the Hasse diagram of the poset has a unique maximal element. This instance includes the important class of NRT-Posets, that is explored, for example, in \cite{park10}.

\begin{myteo}
Let $P$ be a poset with maximal elements $M_P = \{ x_1 , x_2 , \ldots , x_m \}$ such that 
\[ <x_i> \cap <x_j> = \emptyset, \forall \ i \neq j .\]
Denote $ \omega_P (x_i) = l_i $. Then, finding the packing radius of $P$ is equivalent to solving the classic partition problem for the list $S=( l_1, l_2, \ldots , l_m)$.
\end{myteo}

\begin{proof}
Let $(A,B)$ be a partition of $M_P$. Define 
\[ S_1 = ( l_i : x_i \in A) \] 
\[ S_2 = ( l_i : x_i \in B).\]
It is clear that $(S_1, S_2)$ is a partition of $S$. Since
\[ <x_i> \cap <x_j> = \emptyset , \forall \ i \neq j ,\]
we have

\begin{align*}
\omega_P (A) &= \sum_{x \in A} \omega_P (x) \\
&= \sum_{x} \omega_P (x) [ x \in A ] \\
&= \sum_{i} \omega_P (x_i) [x_i \in A] \\
&= \sum_{i} l_i [ l_i \in S_1 ] \\
&= \sum_{l \in S_1} l. \\
\end{align*}

Analogously, 
\[ \omega_P (B) = \sum_{l \in S_2} l ,\]
and therefore, finding a partition $(A,B)$ of $M_P$ that minimizes the maximum between $\omega_P (A)$ and $\omega_P (B)$ is equivalent to finding a partition $(S_1, S_2)$ of $S$ that minimizes the maximum between $\sum_{l \in S_1} l$ e $\sum_{l \in S_2} l$.
\end{proof}

So, to find the packing radius of a poset satisfying the disjoint ideals case we can use all the methods known for the classic partition problem. Given the importance of the discrepancy for the classic partition problem it will be useful to define the discrepancy for the poset case.

\begin{mydef}
Let $P$ be a poset and $(A,B)$ be a partition of $M_P$, the maximal elements of $P$. We define the \textbf{discrepancy} between $A$ and $B$ as
\[ \Delta (A,B) = | \omega_P (A) - \omega_P (B)|, \]
and the \textbf{minimum discrepancy} of $P$ as
\[ \Delta^* (P) = \min_{X \sqcup Y = M_P} \Delta(X,Y) ,\] 
where $X \sqcup Y$ indicates a disjoint union, i.e. $(X,Y)$ is a partition of $M_P$.
\end{mydef}

In the last section we saw that in the classical partition problem what we want to minimize is the discrepancy. For the disjoint ideal case this will then also be true, and we can write the packing radius of a poset, in this case, as the function of its minimum discrepancy.

\begin{myteo}
Let $P$ be a poset of size $n$ with maximal elements $M_P = \{ x_1 , x_2 , \ldots , x_m \}$ such that
\[ <x_i> \cap <x_j> = \emptyset \hspace{0,25 in} \forall \ i \neq j .\]
Then, the packing radius of $P$ is
\[ R(P) = \dfrac{n}{2} + \dfrac{\Delta^* (P)}{2} - 1 .\]
\end{myteo}

\begin{proof}
Let $(A,B)$ be a partition of $M_P$. Then, the following equations are satisfied:
\[ \Delta (A,B) = \omega^\vee_P (A,B) - \min \{ \omega_P (A) , \omega_P (B) \} \]
\[ n = \omega^\vee_P (A,B^v) + \min \{ \omega_P (A) , \omega_P (B) \} .\]

The first one is true by definition. The second one follows from the fact that there is no intersection between the ideals generated by $A$ and $B$, and therefore,
\[n = \omega_P (A) + \omega_P (B) .\]

Taking the sum between both equations and dividing by two we have 
\[ \omega^\vee_P (A,B) = \dfrac{n}{2} + \dfrac{\Delta (A,B)}{2} .\]

But then, by Theorem 2,
\begin{align*}
R(P) &= \min_{A \sqcup B = M_P}  \omega^\vee_P (A,B) - 1 \\
&= \min_{A\sqcup B = M_P} \dfrac{n}{2} + \dfrac{\Delta (A,B)}{2} - 1 \\
&= \dfrac{n}{2} + \dfrac{\Delta^* (P)}{2} - 1 \\
\end{align*}
\end{proof}

\begin{myexe}
Let $H$ be the Hamming poset (anti-chain) of size $n$. The partition problem associated with it is to partition the list with $n$ ones. The minimum discrepancy is then $0$ or $1$ depending on the parity of $n$, i.e.
\[ \Delta^* (H) = [ \text{$n$ is odd}] ,\]
which leads to the classical expression,
\[R(H) = \dfrac{n}{2} + \dfrac{ [ \text{$n$ is odd}] }{2} - 1.\]
\end{myexe}

\begin{myexe}
Let $P$ be a chain of size $n$. The partition problem associated with it is to partition the list $(n)$. The minimum discrepancy is then
\[\Delta^* (P) = n,\]
and therefore,
\begin{align*}
R(P) &= \dfrac{n}{2} + \dfrac{n}{2} - 1 \\
&= n-1, \\
\end{align*}
as found in \cite{fire10}.
\end{myexe}

\begin{figure}[htb]
\centering
\includegraphics[scale=0.3]{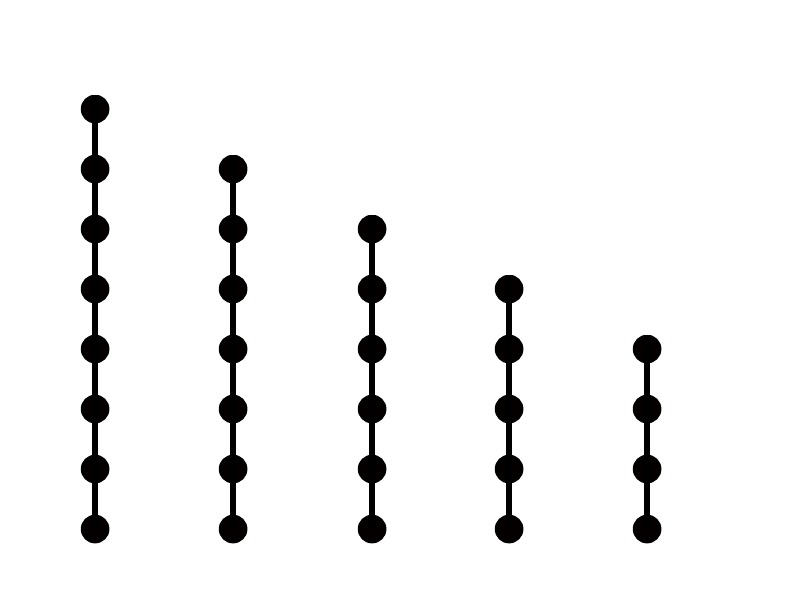}
\caption{Finding the packing radius of this poset, denoted by $P$, is equivalent to solving the partition problem for the list $(8,7,6,5,4)$. We already saw that for this partition $\Delta^* (8,7,6,5,4) = 0$. Thus, since $n = 30$ and $\Delta^* (P) = 0$, the packing radius of the poset is $R(P) = 14$.}
\end{figure}

\subsection{The General Case}

In the general case, finding the packing radius of a poset will not always be equivalent to minimizing the discrepancy as the following figure shows.

\newpage

\begin{figure}[htb]
\centering
\includegraphics[scale=0.3]{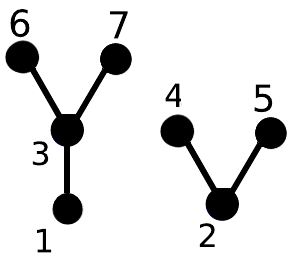}
\caption{In this poset, the partition $ \{ 6,4 \} \{ 7,5 \} $ has discrepancy $\Delta ( \{ 6,4 \}, \{ 7,5 \} ) = 0$ and maximum weight equal to $5$. But the optimal partition is $ \{ 6,7 \} \{ 4,5 \} $ with maximum weight $4$, even though the discrepancy is $1$. }
\end{figure}

We will thus have to consider the possibility of intersections between the ideal generated by maximal elements.

\begin{mydef}
Let $P$ be a poset and $(A,B)$ a partition of $M_P$, the maximal elements of $P$. We define the \textbf{discordancy} between $A$ and $B$ as
\[ \Lambda (A,B) = \Delta (A,B) + | \langle A \rangle \cap \langle B \rangle |  ,\]
and the \textbf{minimum discordancy} of $P$ as
\[ \Lambda^* (P) = \min_{X \sqcup Y = M_P} \Lambda (X,Y) .\] 
\end{mydef}

Note that the discordancy coincides with the discrepancy in the disjoint ideals case.

We now show that the discordancy is what we need to minimize.

\begin{myteo}
Let $P$ be a poset of size $n$. Then, the packing radius of $P$ is
\[ R(P) = \dfrac{n}{2} + \dfrac{\Lambda^* (P)}{2} - 1 .\]
\end{myteo}

\begin{proof}
The proof is analogous to the one given in Theorem 4. The only difference is that in the general case we have
\[ n + | \langle A \rangle \cap \langle B \rangle | = \omega^\vee_P (A,B) + \min \{ \omega_P (A) , \omega_P (B) \} .\]
\end{proof}

\begin{mypro}
Let $P$ be a hierarchical poset of size $n$ and $M_P= \{ x_1 , x_2 , \ldots , x_m \}$ the maximal elements of $P$. Then,
\[ R(P) = n + \dfrac{ [ \text{$m$ is odd} ]}{2} - \dfrac{m}{2} - 1 ,\]
as found in \cite{feli11}.
\end{mypro}

\begin{proof}
If $m=1$:

Then there exists only one partition for $M_P$, $( \{ x_1 \}, \emptyset )$. But $\Lambda  (  \{ x_1 \}, \emptyset ) = n$, and therefore,
\[ \Lambda^* (P) = n .\]

If $m>1$:

The trivial partition $ (M_P , \emptyset) $ is certainly not optimal. Let $(A,B)$ be a non trivial partition. Then 
\[ \langle A  \rangle \cap \langle B \rangle = P-M_P .\]
Thus, the discordancy of any non trivial partition is
\[ \Lambda (A,B) = \Delta (A,B) + | P-M_P | .\]
In this case, minimizing the discordancy is equivalent to minimizing the discrepancy. Since $P-M_P$ is in the intercection of the ideals of any non trivial partition, our problem is equivalent to partitioning the poset $M_P$, a Hamming poset (anti-chain). Therefore, 
\[ \Lambda^* (P) = [ \text{$m$ is odd} ] + | P-M_P | .\]
\end{proof}

\begin{figure}[htb]
\centering
\includegraphics[scale=0.2]{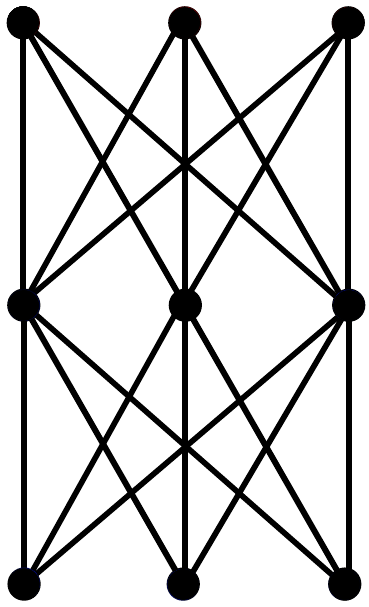}
\caption{Using the notation from our last proposition, the poset $P$ above has parameters $n = 9$, $m = 3$ e $| P - M_P | = 6$. Thus, $ \Lambda^* (P) = 7 $, and therefore, $ R(P) = \dfrac{9+7}{2} - 1 = 7$.  }
\end{figure}

\subsection{The Differencing Method for Posets}

We now aim to generalize the differencing method so that we can use them for any poset. In the disjoint ideals case the only information needed to find the packing radius is the weights of its maximal elements. In the general case we must also consider the intersections between their ideals. To do this, instead of considering numbers we will consider vectors.

\begin{mydef}
Let $P= ( [n] , \preceq )$ be a poset. Given $x \in P$, we denote its \textbf{adjacency vector} by $\hat{x}$ where its coordinates are defined as
\[ \hat{x}_i =  [ i \preceq x ] .\]
\end{mydef}

\begin{myexe}
Let $P = ( [n] , \preceq )$be a poset. Then, its adjacency matrix is
\[  \begin{pmatrix}
\vdots & \vdots & \vdots  & \vdots \\ 
\hat{1}^T& \hat{2}^T & \cdots  & \hat{n}^T \\ 
\vdots & \vdots & \vdots & \vdots
\end{pmatrix} \]
\end{myexe}

Given the set $M_P = \{ x_1, x_2 , \ldots , x_m \}$ of the maximal elements of a poset, we have associated to it a list of adjacency vectors $( \hat{x}_1, \hat{x}_2 , \ldots , \hat{x}_m )$. We will define two operators, the differencing operator($\ominus$), and the associating operator($\oplus$) which will operate on the vectors from the list of adjacency vectors in such a way that differencing two vectors becomes equivalent to committing to place the maximal elements they represent in different subsets, and associating two vectors becomes equivalent to committing to place the maximal elements they represent in the same subset, thus generalizing the differencing method used in the classic partition problem.

\begin{mydef}
Let $X = \{ 0,1,-1,i \}$. The \textbf{differencing and associating operators} are defined by the following tables:

\begin{tabular}{ | c || c | c | c | c | } \hline
$\ominus$ & 0 & 1 & -1 & i \\ \hline \hline
0 & 0 & -1 & 1 & i \\ \hline
1 & 1 & i & 1 & i \\ \hline
-1 & -1 & -1 & i & i \\ \hline
i & i & i & i & i \\ \hline
\end{tabular}
\quad
\begin{tabular}{ | c || c | c | c | c | } \hline
$\oplus$ & 0 & 1 & -1 & i \\ \hline \hline
0 & 0 & 1 & -1 & i \\ \hline
1 & 1 & 1 & i & i \\ \hline
-1 & -1 & i & -1 & i \\ \hline
i & i & i & i & i \\ \hline
\end{tabular}

The value of $x \ominus y$ is found in the $x$ row and $y$ column, for example, \[ 1 \ominus -1 = 1 .\]

For the associating operator the order is immaterial since
\[ x \oplus y = y \oplus x .\]

In the case of two vectors $\hat{x}, \hat{y} \in X^n$, for some $n$, the operators are defined coordinate by coordinate:
\[ ( \hat{x} \oplus \hat{y} )_i = \hat{x}_i \oplus \hat{y}_i \]
\[ ( \hat{x} \ominus \hat{y} )_i = \hat{x}_i \ominus \hat{y}_i .\]
\end{mydef}

The differencing and associating operators behave similarly to addition and subtraction. We define $\oplus x$ to be $0 \oplus x$ and $\ominus x$ to be 
$0 \ominus x$. We now list some of these operators properties:

\begin{mypro}
Let $x,y,z \in \{ 0,1,-1,i \}$. Then:

\begin{enumerate}
\item $x \oplus y = y \oplus x$
\item $(x \oplus y) \oplus z = x \oplus ( y \oplus z )$
\item $0 \oplus x = x \oplus 0 = x$
\item $ \oplus ( x \oplus y ) = \oplus x \oplus y$
\item $\oplus ( x \ominus y ) = \oplus x \ominus y$
\item $ \ominus ( x \oplus y ) = \ominus x \ominus y $
\item $ \ominus ( x \ominus y ) = \ominus x \oplus y $
\end{enumerate}

These properties extend naturally to the vector case.
\end{mypro}

\begin{proof}
This proof follows in a straightforward manner from the definition, but is omitted since it may be quite lengthy.
\end{proof}

Analogously to the classical case we need to associate a partition to every expression involving the differencing and associating operators. We begin with simple expressions.

\begin{mydef}
Let $( \hat{x}_1, \hat{x}_2 , \ldots , \hat{x}_m )$ be the list of adjacency vectors associated with the maximal elements of a poset. A \textbf{simple expression} involving the elements of this list and the operators $\oplus$ and $\ominus$ is a sequence of the form
\[ *_{1} \hat{x}_{k_1} *_{2} \hat{x}_{k_2} \ldots *_n \hat{x}_{k_n} ,\]
where $k_i \in [m]$ and $*_i = \oplus$ or $*_i = \ominus$.
\end{mydef}

The \textbf{partition associated} with the simple expression is the partition $(A,B)$ defined as follows:
\[ A = \{ x_{k_i} : *_i = \oplus \} \]
\[ B  = \{ x_{k_i} : *_i = \ominus \} .\]

$A$ is called the \textbf{primary set} and $B$ is called the \textbf{secondary set}.

We now extend the definition to any expression.

\begin{mydef}
The partition associated with any expression is the partition associated with the simple expression obtained by applying the properties of Proposition 3 to the original expression.
\end{mydef}

We make the following notation abuse: we sometimes denote an expression by the vector we would obtain by following the calculations on the expression, i.e. we might denote the expression $\hat{x} \oplus \hat{y}$ by the vector $\hat{v}= \hat{x} \oplus \hat{y}$. The $k$ coordinate of an expression is then the $k$ coordinate of the vector which is a result of the expression if it were calculated.

We denote the primary set of an expression $\hat{v}$ by $Pri(\hat{v})$, and its secondary set by $Sec(\hat{v})$.

\begin{myexe}
Let $P$ be a poset with  maximal elements $M_P = \{ x_1 , x_2, x_3 , x_4 \}$. The adjacency vectors associate with $M_P$ are $( \hat{x}_1, \hat{x}_2 , \hat{x}_3 , \hat{x}_4 )$. The simple form of the expression 
\[ ( \hat{x}_1 \ominus \hat{x}_2 ) \ominus ( \hat{x}_3 \ominus \hat{x}_4 ) \]
is
\[ \oplus \hat{x}_1 \ominus \hat{x}_2 \ominus \hat{x}_3 \oplus \hat{x}_4 \]
which is associated to the partition $( \{ x_1,x_4 \}, \{ x_2, x_3 \} )$ where $\{ x_1,x_4 \}$ is the primary set and $\{ x_2, x_3 \}$ is the secondary set.
\end{myexe}

The primary and secondary sets have the following properties.

\begin{mypro}
Let $P$ be a poset with maximal elements $M_P = \{ x_1, x_2 , \ldots , x_m \}$, and associated adjacency vectors $( \hat{x}_1, \hat{x}_2 , \ldots , \hat{x}_m )$. If $\hat{v}$ and $\hat{w}$ are two expressions using elements from the adjacency list, then:
\begin{enumerate}
\item $Pri ( \hat{v} \oplus \hat{w} )= Pri ( \hat{v} ) \cup Pri ( \hat{w} )$
\item $Sec ( \hat{v} \oplus \hat{w} )= Sec ( \hat{v} ) \cup Sec ( \hat{w} )$
\item $Pri ( \hat{v} \ominus \hat{w} )= Pri ( \hat{v} ) \cup Sec ( \hat{w} )$
\item $Sec ( \hat{v} \ominus \hat{w} )= Sec ( \hat{v} ) \cup Pri ( \hat{w} )$
\end{enumerate}
\end{mypro}

\begin{proof}
Just put $\hat{v}$ and $\hat{w}$ in their simple forms and note that, by properties $4$ and $5$ of Proposition $3$, $\oplus$ does not alter the operators in front of the vectors and, by properties $6$ and $7$ of Proposition $3$, $\ominus$ switches all the operators in front of the vectors of $\hat{w}$.
\end{proof}

We can now prove the following:

\begin{mylem}
Let $P$ be a poset of size $n$, $M_P = \{ x_1, x_2 , \ldots , x_m \}$ be the maximal elements of $P$, $( \hat{x}_1, \hat{x}_2 , \ldots , \hat{x}_m )$ be the adjacency vectors associated with $M_P$, and $\hat{v}$ and $\hat{w}$ be two expressions using the adjacency vectors. 
Suppose that for every $k\in [n]$:
\begin{enumerate}
\item $\hat{v}_k = 0 \Leftrightarrow k \notin \langle Pri(\hat{v}) \rangle \cup \langle Sec(\hat{v}) \rangle$
\item $\hat{v}_k = 1 \Leftrightarrow k \in \langle Pri(\hat{v}) \rangle - \langle Sec(\hat{v}) \rangle$
\item $\hat{v}_k = -1 \Leftrightarrow k \in \langle Sec(\hat{v}) \rangle - \langle Pri(\hat{v}) \rangle$
\item $\hat{v}_k = i \Leftrightarrow k \in \langle Pri(\hat{v}) \rangle \cup \langle Sec(\hat{v}) \rangle$
\end{enumerate}
and that these properties are still true when we substitute $\hat{v}$ for $\hat{w}$. Then, the properties listed are still true when we substitute $\hat{v}$ for $\hat{v} \oplus \hat{w}$ or $\hat{v} \ominus \hat{w}$.
\end{mylem}

\begin{proof}
The proof consists in separating in all possible cases and using the properties from Proposition 4. Since there are four possible values for $\hat{v}_k$ and for $\hat{w}_k$ and there are two operations, the total number of cases is thirty-two. We will only show one case as an example. The proof for the other cases is completely analogous.

\textbf{Associating case for $\hat{v}_k = 1$ and $\hat{w}_k = -1$:}

In this case,
\[ \hat{v}_k \oplus \hat{w}_k = i .\]
Our two hypothesis tell us that 
\[k \in \langle Pri(\hat{v}) \rangle - \langle Sec(\hat{v}) \rangle \] 
and
\[ k \in \langle Sec(\hat{w}) \rangle - \langle Pri(\hat{w}) \rangle .\]
But then,
\[ k \in \langle Pri(\hat{v}) \rangle \subseteq \langle Pri(\hat{v}_k \oplus \hat{w}_k) \rangle \]
and
\[ k \in \langle Sec(\hat{w}) \rangle \subseteq \langle Sec(\hat{v}_k \oplus \hat{w}_k) \rangle ,\]
and therefore,
\[ k \in \langle Pri(\hat{v}_k \oplus \hat{w}_k) \rangle \cup \langle Sec(\hat{v}_k \oplus \hat{w}_k) \rangle .\]
\end{proof}

Since the list of vectors $( \hat{x}_1, \hat{x}_2 , \ldots , \hat{x}_m )$ satisfies the conditions of our lemma the differencing and associating operators work in the way we intended them to. In the same way we built a tree for the classic partition problem, we can now build a tree for the poset partition problem where in each node we will have a list of vectors, left branches substitute two vectors for their difference, and right branches substitute them for their associate. By our last lemma the terminal nodes, consisting of single vectors, must have the information necessary to calculate the discordancy of the partition associated with each vector. To extract this information we need the following:

\begin{mydef}
Let $\hat{v} \in \{0,1,-1,i \}^n$. We define the sum of entries function as
\[ S (\hat{v}) = \sum_{k = 1}^{n} \hat{v}_k \]
where $i$ is treated formally as if it where the imaginary unit.
\end{mydef}

\begin{myexe}
If $\hat{v} = (1, 1, -1 , i , 1, i)$ then
\[ S(\hat{v}) = 2 +2i .\]
\end{myexe}

We can therefore refer to the real part of $S(\hat{v})$, denoted as $ \Re ( S(\hat{v}) ) $, and the imaginary part, denoted as $ \Im ( S(\hat{v}) ) $. We can now find explicitly the discordancy of the terminal nodes of our tree.

\begin{myteo}
Let $P$ be a poset with maximal elements $M_P = \{ x_1, x_2 , \ldots , x_m \} $, $( \hat{x}_1, \hat{x}_2 , \ldots , \hat{x}_m )$ be the adjacency vectors associated with $M_P$, and $\hat{v}$ be an expression using all the adjacency vectors exactly once. If we denote by $(A,B)$ the partition of $M_P$ associated to $\hat{v}$, then,
\[ \Delta (A,B) = | \Re ( S(\hat{v}) ) | \] 
and
\[ | \langle A \rangle \cap \langle B \rangle |  = \Im ( S(\hat{v}) ) .\]
Thus,
\[ \Lambda (A,B) = | \Re ( S(\hat{v}) ) | + \Im ( S(\hat{v}) ) .\]
\end{myteo}

\begin{proof}
Without loss of generality, suppose $A = Pri (\hat{v})$ and $B = Sec (\hat{v})$. Then, since the adjacency vectors of $M_P$ satisfy the conditions of Lemma 5, it follow that,
\[ | \langle A \rangle | = \sum_{k=1}^n [v_k = 1] + \sum_{k=1}^n [v_k = i] \]
and
\[ | \langle B \rangle | = \sum_{k=1}^n [v_k = -1] + \sum_{k=1}^n [v_k = i] .\]
Thus,
\begin{align*}
\Delta (A,B) &= \left | \sum_{k=1}^n [v_k = 1] - \sum_{k=1}^n [v_k = -1] \right |  \\
& = \left |  \sum_{k=1}^n v_k [v_k \neq i] \right | \\
& = \left | \Re ( S(\hat{v}) ) \right | .
\end{align*}

But by Lemma 5 we also have
\begin{align*}
| \langle A \rangle \cap \langle B \rangle | &=  \sum_{k=1}^n [v_k = i] \\
&= \Im ( S(\hat{v}).
\end{align*}
\end{proof}

We can now build a tree similar to the one in the CKK algorithm. We will make one modification, nonetheless, we will substitute the lists of vectors for matrices.

\begin{mydef}
Let $P$ be a poset of size $n$ with maximal elements $M_P = \{ x_1, x_2 , \ldots , x_m \}$, and $( \hat{x}_1, \hat{x}_2 , \ldots , \hat{x}_m )$ be the adjacency vectors associated with $M_P$. We define the \textbf{radius matrix} of $P$ as
\[  \begin{pmatrix}
\vdots & \vdots & \vdots  & \vdots \\ 
\hat{x}_1^T& \hat{x}_2^T & \cdots  & \hat{x}_m^T \\ 
\vdots & \vdots & \vdots & \vdots
\end{pmatrix} .\]
\end{mydef}

Note that this matrix can be obtained from the adjacency matrix of $P$ by removing the columns which do not correspond to maximal elements.

By Theorem 6, the packing radius of a poset is completely determined by its radius matrix. The packing radius can therefore be seen as the property of a matrix.

We will now extend the definition of discordancy and packing radius to matrices.

We begin by the discordancy of a vector.

\begin{mydef}
Let $\hat{v}$ be a vector with elements in $\{ 0,1,-1,i \}$. The \textbf{discordancy of $\hat{v}$} is defined as 
\[ \Lambda ( \hat{v}) =  | \Re ( S(\hat{v}) ) | + \Im ( S(\hat{v}) ) .\]
\end{mydef}

The discordancy of a matrix will be given recursively.

\begin{mydef}
Let
\[ M = \begin{pmatrix}
\vdots & \vdots & \vdots  & \vdots \\ 
\hat{x}_1^T& \hat{x}_2^T & \cdots  & \hat{x}_m^T \\ 
\vdots & \vdots & \vdots & \vdots
\end{pmatrix} \]
be a matrix with elements in $\{ 0,1,-1,i \}$. 

We define $M^{\oplus}_{j,k}$ as the matrix
\[ \begin{pmatrix}
\vdots & \vdots & \vdots  & \vdots & \vdots & \vdots \\ 
\hat{x}_1^T& \hat{x}_2^T & \cdots & \hat{x}_j^T \oplus \hat{x}_k^T &\cdots  & \hat{x}_m^T \\ 
\vdots & \vdots & \vdots & \vdots & \vdots & \vdots
\end{pmatrix}  
,\]
i.e. the matrix $M$ after substituting columns $j$ and $k$ with the column associating their corresponding vectors.

The matrix $M^{\ominus}_{j,k}$ is defined analogously.
\end{mydef}

Note that both $M^{\oplus}_{j,k}$ and $M^{\ominus}_{j,k}$ have one less column than $M$.

\begin{mydef}
Let $M$ be a matrix with elements in $\{ 0,1,-1,i \}$. The \textbf{minimum discordancy of $M$} is defined as
\[ \Lambda^* ( M ) = \min \left \{
\Lambda^* ( M^{\oplus}_{j,k} )
,
\Lambda^*   M^{\ominus}_{j,k}
  \right \}
,\]
where the choice of $j$ and $k$ is irrelevant as long as they are different. 

The minimum discordancy of a vector is defined as its discordancy.
\end{mydef}

We now define the packing radius of a matrix.

\begin{mydef}
Let $M$ be a matrix with elements in $\{ 0,1,-1,i \}$ and $n$ be the number of rows in $M$ where there exists at least one element different than $0$. The \textbf{packing radius of $M$} is defined as
\[ R(M) = \dfrac{n}{2} + \dfrac{\Lambda^* (M)}{2}-1 .\]
\end{mydef}

Some direct properties of the packing radius of a matrix is that it does not change under row or column permutations.

We can now rephrase Theorem 6 in terms of the radius matrix of a poset.

\begin{myteo}
Let $P$ be a poset and $M$ its radius matrix. Then $R(P) = R(M)$.
\end{myteo}

\begin{proof}
By Theorem 6, $\Lambda^*(P) = \Lambda^*(M)$.

Let $n$ be the size of $P$. Thus, $M$ has $n$ rows and all of them must have at least one element different from $0$ since $P$ is a poset. Therefore, the packing radius of $M$ is
\[ R(M) = \dfrac{n}{2} + \dfrac{\Lambda^* (P)}{2}-1 .\]
\end{proof}

To construct a searching tree using the differencing method for posets we need to have a criterion for choosing the vectors to differenciate or associate. In the classical problem the best criterion was LDM. We present now a generalization of the LDM criterion and, despite the inexistence of different criteria, conjecture that it should perform well in the same context as in the classical partition problem. We call it the Poset LDM (PLDM) criterion.Our first vector, $\hat{v}$, will be the one that maximizes $\Lambda^*(\hat{v})$, and our second one, $\hat{w}$, will then be the one that minimizes $\Lambda^*(\hat{v} \ominus \hat{w})$. In our examples we will always list these two vectors in the first two columns.

\begin{myexe}
Lets find the packing radius of a poset $P$ with adjacency matrix
\[
\begin{pmatrix}
1 & 0 & 0 & 1 & 1 & 0 & 0 \\ 
0 & 1 & 0 & 0 & 1 & 1 & 1 \\ 
0 & 0 & 1 & 1 & 1 & 0 & 0 \\ 
0 & 0 & 0 & 1 & 0 & 0 & 0 \\ 
0 & 0 & 0 & 0 & 1 & 0 & 0 \\ 
0 & 0 & 0 & 0 & 0 & 1 & 0 \\ 
0 & 0 & 0 & 0 & 0 & 0 & 1
\end{pmatrix}
\]

The maximal elements of $P$ are $M_P = \{ 4 , 5 , 6 , 7 \}$, and therefore, its radius matrix is
\[
\begin{pmatrix}
1 & 0 & 1 & 0 \\ 
1 & 1 & 0 & 1 \\ 
1 & 0 & 1 & 0 \\ 
0 & 0 & 1 & 0 \\
1 & 0 & 0 & 0 \\ 
0 & 1 & 0 & 0 \\ 
0 & 0 & 0 & 1
\end{pmatrix}
\]

\begin{figure}[htb]
\centering
\includegraphics[scale=0.2]{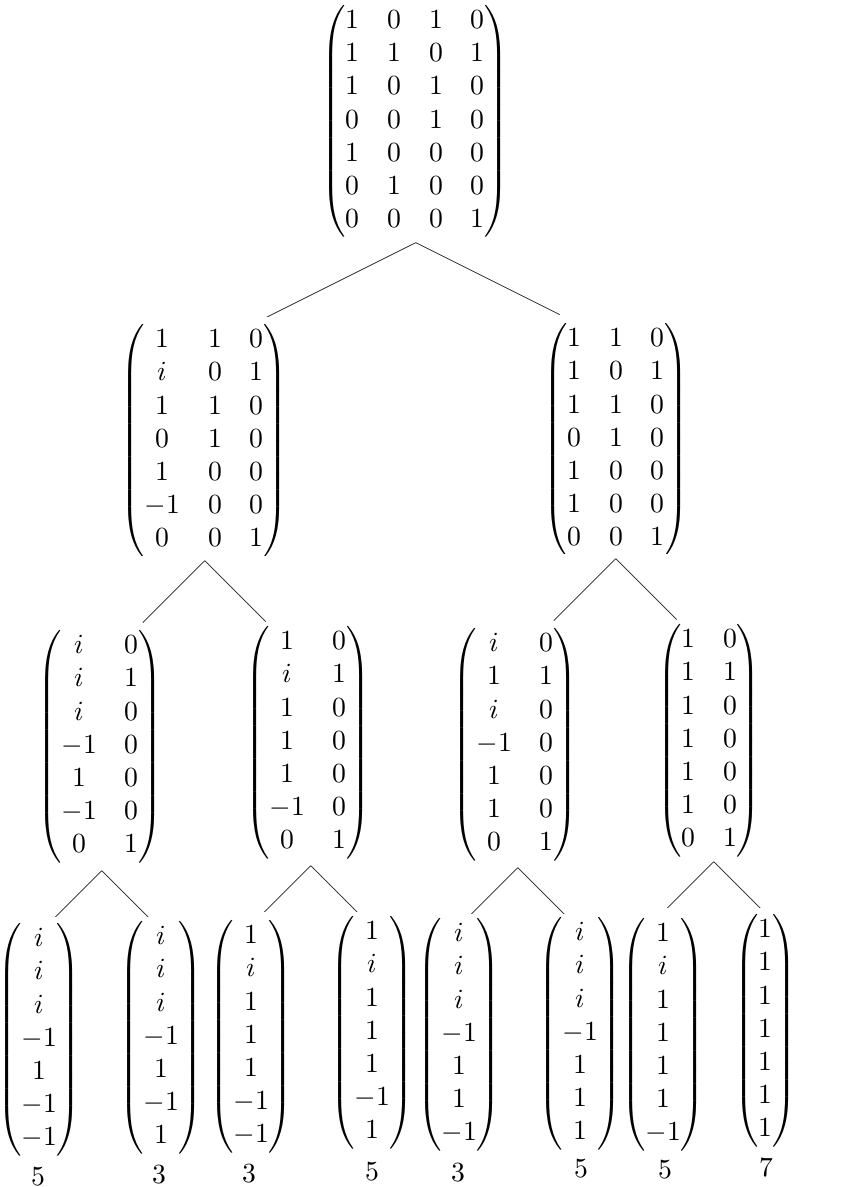}
\caption{ Tree }
\end{figure}

Figure 6 shows the tree due to the differencing method for posets applied to the radius matrix of $P$. From it we conclude that $\Lambda^* (P) = 3$. Thus, $R(P) = 4$.
\end{myexe}

\subsection{Pruning for the Poset Case}

We can shorten our notation in the following way. Note that if a matrix in the node of the tree has an entry with value $i$, then all terminal nodes deriving from it has the value $i$ in that row. This is due to the fact that $i$ is always preserved by the differencing and associating operations. We can therefore omit any line where an $i$ appears and add a counter in front of the matrix. More precisely, we substitute a matrix $L$ for a pair $(\alpha, M)$, called a \textbf{number-matrix}, where $\alpha$ is the number of rows in $L$ in which an $i$ appears and $M$ is obtained from the matrix $L$ by removing those $\alpha$ rows. Then, instead of applying the differencing method on $L$, we apply it on $M$ but we remove again $\beta$ rows that have an $i$ on them and end up with the number-matrix $(\alpha + \beta , N)$. We do this until we get to a terminal node $(\gamma , \hat{v})$, and because of Theorem 6, the discordancy of the terminal node will be $\gamma + S(\hat{v})$.

We denote a number-matrix $(\alpha, M)$ by $\alpha M$ omitting the $\alpha$ if it is zero, taking care to not confuse it with a number multiplying a matrix.

\begin{myexe}
On Figure 6, the first matrix on the left branch is
\[
\begin{pmatrix}
1 & 1 & 0 \\ 
i & 0 & 1 \\ 
1 & 1 & 0 \\ 
0 & 1 & 0 \\ 
1 & 0 & 0 \\
-1 & 0 & 0 \\ 
0 & 0 & 1 
\end{pmatrix}
.\]
Using the number-matrix notation we would substitute it with
\[
1 \begin{pmatrix}
1 & 1 & 1 \\ 
1 & 1 & 0 \\ 
0 & 1 & 0 \\ 
1 & 0 & 0 \\
-1 & 0 & 0 \\ 
0 & 0 & 1 
\end{pmatrix}
.\]
\end{myexe}

We now show how to prune the tree in the poset case.

\begin{myteo}
Let 
\[ \alpha M = \alpha \begin{pmatrix}
\vdots & \vdots & \vdots  & \vdots \\ 
\hat{w}_1^T& \hat{w}_2^T & \cdots  & \hat{w}_k^T \\ 
\vdots & \vdots & \vdots & \vdots
\end{pmatrix} \]
be a number-matrix which is the node of a tree resulting from the differencing method. Then,
\[ \Lambda^* (\alpha M) \geq \alpha + \Delta^* ( S(\hat{w}_1) , S(\hat{w}_2) , \ldots , S(\hat{w}_k) ) \]
\end{myteo}

\begin{proof}
It is simple to show that
\[ \Lambda^* (\alpha M) = \alpha + \Lambda^* (M) .\]

But, $\Lambda^* (M) \geq \Delta^* ( S(\hat{w}_1) , S(\hat{w}_2) , \ldots , S(\hat{w}_k) )$ since the appearances of $i's$ only add to $\Lambda^* (M)$.
\end{proof}

With this theorem we have the following pruning method:
After reaching a terminal node $\beta \hat{v}$ we can prune every number-matrix
\[ \alpha M = \alpha \begin{pmatrix}
\vdots & \vdots & \vdots  & \vdots \\ 
\hat{w}_1^T & \hat{w}_2^T & \cdots  & \hat{w}_k^T \\ 
\vdots & \vdots & \vdots & \vdots
\end{pmatrix} \]
such that 
\[ \alpha + \Delta^* (S(\hat{w}_1) , S(\hat{w}_2) , \ldots , S(\hat{w}_k)) - 1 \geq \beta + S(\hat{v}).\]

\begin{myexe}
Let $P$ be the poset with adjacency matrix
\[
\begin{pmatrix}
1 & 1 & 1 & 1 & 1 & 1 & 1 \\ 
0 & 1 & 1 & 1 & 1 & 1 & 1 \\ 
0 & 0 & 1 & 1 & 1 & 1 & 0 \\ 
0 & 0 & 0 & 1 & 0 & 0 & 0 \\ 
0 & 0 & 0 & 0 & 1 & 0 & 0 \\ 
0 & 0 & 0 & 0 & 0 & 1 & 0 \\ 
0 & 0 & 0 & 0 & 0 & 0 & 1
\end{pmatrix}
.\]

The maximal elements of $P$ are $M_P = \{ 4 , 5 , 6 , 7 \}$. Thus, its radius matrix is
\[
\begin{pmatrix}
1 & 1 & 1 & 1 \\ 
1 & 1 & 1 & 1 \\ 
1 & 1 & 1 & 0 \\ 
1 & 0 & 0 & 0 \\
0 & 1 & 0 & 0 \\ 
0 & 0 & 1 & 0 \\ 
0 & 0 & 0 & 1
\end{pmatrix}
.\]

\begin{figure}[htb]
\centering
\includegraphics[scale=0.2]{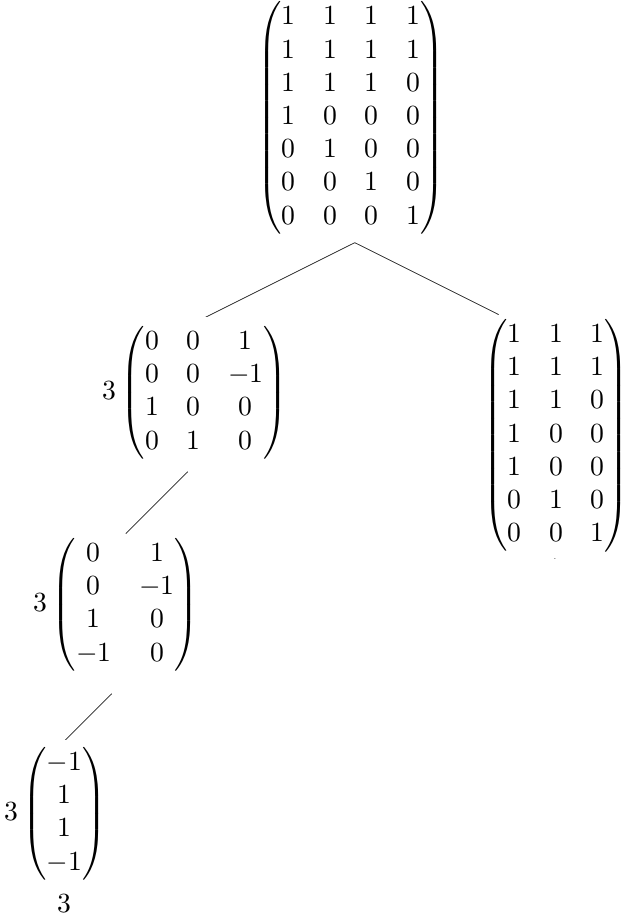}
\caption{ Pruned tree }
\end{figure}

Figure 7 shows the pruned tree using the method just described. We first reached a terminal node with discordancy $3$. That pruned all the branches on the left side since the number outside the first matrix to the left is $3$. We are left only with the matrix 
\[
\begin{pmatrix}
1 & 1 & 1 \\ 
1 & 1 & 1 \\ 
1 & 1 & 0 \\ 
1 & 0 & 0 \\
1 & 0 & 0 \\ 
0 & 1 & 0 \\ 
0 & 0 & 1
\end{pmatrix}
.\]
But the sum of each one of its columns gives us the list $(5,4,3)$. Since $\Delta^*(5,4,3) = 2$ and $\Lambda^* (P)$ must be odd, the smallest value possible for the discrepancy resulting from a terminal node of this matrix is $3$. Therefore, $\Lambda^* (P) = 3$, and $ R(P) = 4$.
\end{myexe}

\section{The Packing Vector}

We now come back to the problem of finding the packing radius of a linear code. To do this we saw that we have to find its packing vector, the code-word with minimum packing radius. One way to do this would be to calculate the packing radius of each code-word, but as we have seen that would be a big problem since we would have to solve a poset partition problem for each code-word. What we need is a way to compare the packing radius of code-words without calculating them. In other words, given two posets we want to compare their packing radius without explicitly determining them.

The simplest way to compare the packing radius of two posets is using the following well known inequality:
\[ \left \lfloor \dfrac{d_P (C)-1}{2} \right \rfloor \leq R_{d_P} (C) \leq d_P (C) -1 .\]

\begin{mypro}
Let $P$ and $Q$ be two posets of size $n$ and $m$ respectively. Then,
\[ n \leq \dfrac{m}{2} + \dfrac{[\text{$m$ is odd}]}{2}  \Rightarrow R(P) \leq R(Q).\]
\end{mypro}

\begin{proof}
Suppose
\[ n \leq \dfrac{m}{2} + \dfrac{[\text{$m$ is odd}]}{2} .\]
Then, by the above inequality
\[R(P) + 1 \leq n ,\]
and therefore,
\[ R(P) \leq \dfrac{m}{2} + \dfrac{[\text{$m$ is odd}]}{2} -1 .\]
But the above inequality also tells us that
\[\dfrac{m}{2} + \dfrac{[\text{$m$ is odd}]}{2} -1 \leq R(Q) \]
from where the result follows.
\end{proof}

We can therefore eliminate code-words which are much greater than the minimum weight of the code.

For the next result we will need the following definition:

\begin{mydef}
Let $P=(A, \preceq_P)$ and $Q=(B, \preceq_Q)$ be two posets. We say $P$ is a subposet of $Q$, denoted by $P \subseteq Q$, if
\[ A \subseteq B \]
and
\[ x \preceq_Q y \Rightarrow x \preceq_P y,  \hspace{10pt} \forall x,y \in A.\]
\end{mydef}

\begin{myteo}
Let $P$ and $Q$ be posets. Suppose there is another poset $P_2 \subseteq Q$ isomorphic to $P$. Then,
\[R(P) \leq R(Q) .\]
\end{myteo}

\begin{proof}
Let $(A,B)$ be a partition of $Q$. Then, $(A \cap P_2 , B \cap P_2)$ is a partition of $P_2$ with the property that 
\[ \omega_{P_2} (A \cap P_2) \leq \omega_Q (A) \]
and
\[ \omega_{P_2} (B \cap P_2) \leq \omega_Q (B) .\]
Thus, $R(P_2) \leq R(Q)$. Now, since the packing radius is a property of the poset, it is invariant under isomorphism, and therefore, $R(P_2) = R(P)$.
\end{proof}

With this theorem we can calculate the packing radius of a poset linear code when the poset is hierarchical.

\begin{mypro}
Let $P$ be a hierarchical poset, $C \subseteq \mathbb{F}_q^n$ be a $P$-linear code and $c \in C$ be a minimum weight code-word. If $I_c$ is the ideal generated by the support of $c$ and $M_{I_c} = \{ x_1 , x_2, \ldots , x_m \}$ are the maximal elements of $I_c$, then
\[ R_P(C) = \omega_P (C) + \dfrac{[\text{$m$ is odd}]}{2} - \dfrac{m}{2} -1 .\]
\end{mypro}

\begin{proof}
Let $I$ and $J$ be two ideals in $P$. Since $P$ is hierarchical, then so are $I$ and $J$. Suppose $ |I| \leq |J|$. Then there are two possibilities: $I \subseteq J$ or the maximal elements of $I$ are in the same level of the hierarchy. In the first case it is clear that $R(I) \leq R(J)$. In the second case, $I$ must have less maximal elements than $J$, and therefore, there exists an isomorphism between $I$ and a subset of $J$. A direct application of Theorem 9 gives us $R(I) \leq R(J)$.

Thus, any minimal weight vector will be the packing vector and Proposition 2 yields our result
\end{proof}

In the last section, we saw that the packing radius of a poset is a property of its radius matrix. It must, therefore, be possible to compare the packing radius of two posets by comparing their radius matrices.

Our next result shows how to transform our poset into a simpler one without modifying its packing radius.

\begin{myteo}
Let $P$ be a poset of size $n$. Then, there exists a poset $Q$ of size $n$ which we call the \textbf{standard form of $P$} such that:
\begin{itemize}
\item $R(Q) = R(P)$.
\item Every element of $Q$ is either maximal or minimal.
\end{itemize}
\end{myteo}

\begin{proof}
Let $M$ be the radius matrix of $P$ and $A$ its adjacency matrix. Note that $M$ is a sub-matrix of $A$. We define $Q$ as the poset whose adjacency matrix coincides with $M$ in the corresponding columns and is zero elsewhere (other than the main diagonal).Since $Q$ has the same radius matrix as $P$, $R(Q) = R(P)$. Also, since the only relations in $Q$ involving different elements are in the sub-matrix $M$ whose columns correspond to maximal elements, every element of $Q$ is maximal or minimal.
\end{proof}

\begin{myexe}
Let $P$ be the poset with adjacency matrix
\[A_P =  \begin{pmatrix}
1 & 0 & 0 & 1 & 1 & 1 & 1 & 1  \\
0 & 1 & 0 & 1 & 0 & 0 & 1 & 0  \\ 
0 & 0 & 1 & 0 & 1 & 0 & 0 & 1  \\
0 & 0 & 0 & 1 & 0 & 0 & 1 & 0  \\
0 & 0 & 0 & 0 & 1 & 0 & 0 & 1  \\
0 & 0 & 0 & 0 & 0 & 1 & 0 & 0  \\ 
0 & 0 & 0 & 0 & 0 & 0 & 1 & 0  \\
0 & 0 & 0 & 0 & 0 & 0 & 0 & 1  
\end{pmatrix}
.\]
To find the maximal elements of $P$ we search for the rows in which only one $1$ appears. In this case, rows $6$, $7$ and $8$. The radius matrix of $P$ is, therefore, the matrix whose columns are the last three columns of $A_P$. Zeroing all the entries of $A_P$ other than those in the main diagonal and in the last four columns we obtain the following matrix
\[A_Q =  \begin{pmatrix}
1 & 0 & 0 & 0 & 0 & 1 & 1 & 1  \\
0 & 1 & 0 & 0 & 0 & 0 & 1 & 0  \\ 
0 & 0 & 1 & 0 & 0 & 0 & 0 & 1  \\
0 & 0 & 0 & 1 & 0 & 0 & 1 & 0  \\
0 & 0 & 0 & 0 & 1 & 0 & 0 & 1  \\
0 & 0 & 0 & 0 & 0 & 1 & 0 & 0  \\ 
0 & 0 & 0 & 0 & 0 & 0 & 1 & 0  \\
0 & 0 & 0 & 0 & 0 & 0 & 0 & 1  
\end{pmatrix}
,\]
which is the adjacency matrix for the standard form of $P$.

\begin{figure}[htb]
\centering
\includegraphics[scale=0.2]{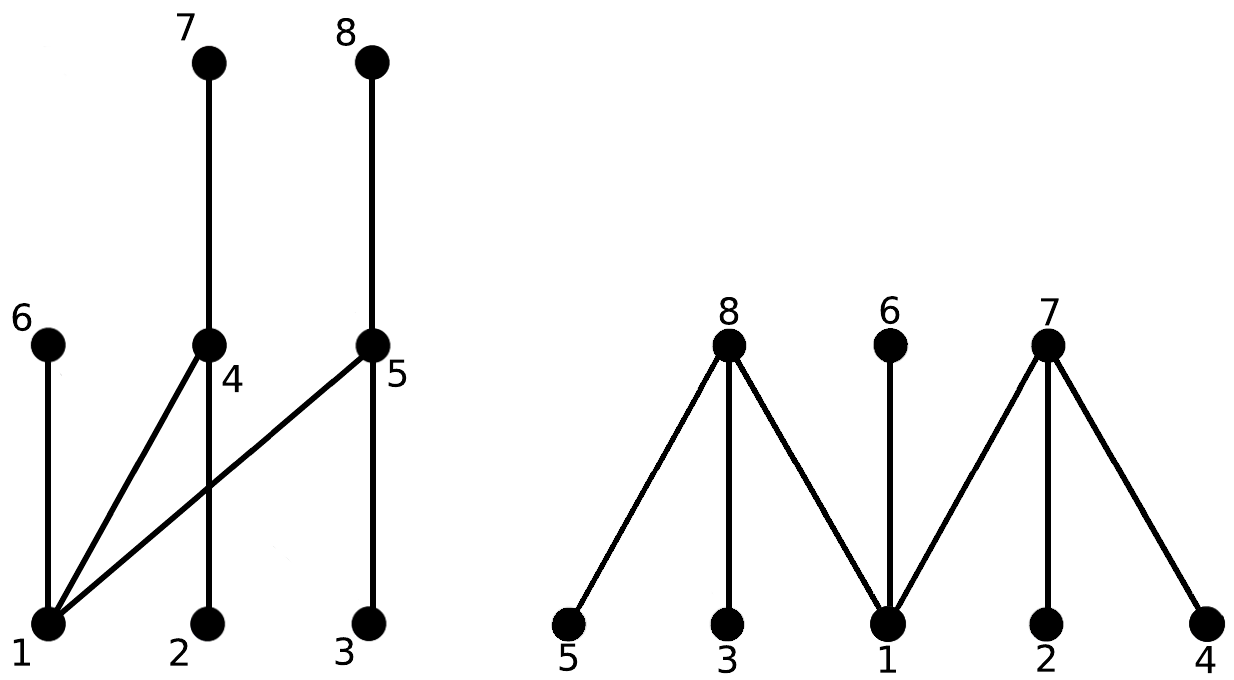}
\caption{Posets $P$, on the left, and its standard form, on the right.}
\end{figure}

\end{myexe}

By joining Theorem 10 with Theorem 9 we can compare the packing radius of two posets by comparing there standard forms.

We now generalize the notion of the radius matrix.

\begin{mydef}
Let $P$ be a poset. We call a matrix $M$ an \textbf{extended radius matrix of $P$} or an ER-matrix of $P$ if $R (M) = R (P)$.
\end{mydef}

Note that the radius matrix of a poset $P$ is an ER-matrix of it.

Given an RE-matrix of a poset $P$ we show how to construct others.

\begin{mypro}
Let $P$ be a poset and $M$ be an RE-matrix of $P$. Then the following operations on $M$ preserve the fact that it is an ER-matrix of $P$:
\begin{enumerate}
\item Swapping two rows.
\item Swapping two columns.
\item Adding or removing a null row, i.e. a row only with zeros.
\item Adding or removing a column whose support is contained in the support of a column of $M$.
\end{enumerate}
\end{mypro}

\begin{proof}
Note that in all cases the number of non-null rows remains the same, so that it is sufficient to show that the minimum discordancy must remain equal.

The first three properties are true since these operations do not alter the results of the differencing method in any way.

The fourth property is true for the following reason: Let $\hat{v}$ be a column of $M$. If we add a column $\hat{w}$ to $M$ such that $supp(\hat{w}) \subseteq supp(\hat{v})$, then $\hat{v} \oplus \hat{w} = \hat{v}$ and $\hat{v} \ominus \hat{w}$ is equal to $\hat{v}$ but with $i$'s in the rows in which $\hat{w}$ has value 1. Thus, associating $\hat{v}$ with $\hat{w}$ leads to terminal nodes with smaller discordancies, but this sub-tree will be the same one as that of $M$.
\end{proof}

An interesting fact is that by the fourth property of the last proposition, the adjacency matrix $A$ of a poset $P$ is an ER-matrix of $P$, i.e. $R(P) = R(A)$.

\begin{myexe}
Let $P$ and $Q$ be two posets with adjacency matrices
\[ A_P = \begin{pmatrix}
1 & 1 & 1 & 1 & 1 \\ 
0 & 1 & 0 & 1 & 0 \\ 
0 & 0 & 1 & 1 & 1 \\ 
0 & 0 & 0 & 1 & 0 \\
0 & 0 & 0 & 0 & 1 
\end{pmatrix}
\]
and
\[ A_Q = \begin{pmatrix}
1 & 0 & 1 & 1 & 1 \\ 
0 & 1 & 0 & 1 & 1 \\ 
0 & 0 & 1 & 0 & 1 \\ 
0 & 0 & 0 & 1 & 0 \\
0 & 0 & 0 & 0 & 1 
\end{pmatrix}
\]
respectively.

Lets show that $A_P$ is an ER-matrix of $Q$.

Starting by the matrix $A_Q$, using property $4$ of Proposition 7, we can remove the first three columns. Thus,
\[  \begin{pmatrix}
1 & 1 \\ 
1 & 1 \\ 
0 & 1 \\ 
1 & 0 \\
0 & 1 
\end{pmatrix}
\]
is an ER-matrix of $Q$.

Using properties $1$ and $2$ we swap both columns, rows $2$ and $3$, and rows $4$ and $5$, resulting in
\[  \begin{pmatrix}
1 & 1 \\ 
1 & 0 \\ 
1 & 1 \\ 
1 & 0 \\
0 & 1 
\end{pmatrix}
.\]
But using property $4$ again we can get to the matrix $A_P$ since the support of the first three columns is contained in the first column of the above matrix.

Thus, $R(P) = R(Q)$.
\end{myexe}

In the next result we show how to compare the packing radius of two posets using their ER-matrices.

\begin{myteo}
Let $P$ and $Q$ be two posets with ER-matrices $A_P$ and $A_Q$ respectively. Then,
\[ supp(A_P) \subseteq supp(A_Q) \Rightarrow R(P) \leq R(Q) .\]
\end{myteo}

\begin{proof}
We will prove the equivalent statement that if we change any entry of a matrix to the value $0$, its packing radius will not increase.

Let $M$ be a matrix and $j$ and $k$ a row and a column, respectively, of $M$. Let $N$ be the matrix that coincides with $M$ everywhere except on the $j$th row and the $k$th column where it assumes the value $0$.

Note that the most that a single row can contribute to the packing radius of the matrix is $1$.

If the $j$th row of $N$ is null, than certainly, $R(N) \leq R(M)$.

If the $j$th row of $N$ is non-null, than it is a matter of comparing the minimum discordancies of $M$ and $N$. By looking at the properties of the associating and differencing operations it is clear that a $0$ will always lead to better results, in terms of lower discordancies, than other values.
\end{proof}

\begin{myexe}
Let $P$ and $Q$ be two posets with adjacency matrices
\[ A_P = \begin{pmatrix}
1 & 0 & 0 & 0 & 1 & 0 & 0 \\ 
0 & 1 & 1 & 1 & 1 & 1 & 1 \\ 
0 & 0 & 1 & 0 & 0 & 0 & 1 \\ 
0 & 0 & 0 & 1 & 1 & 0 & 1 \\
0 & 0 & 0 & 0 & 1 & 0 & 0 \\ 
0 & 0 & 0 & 0 & 0 & 1 & 0 \\ 
0 & 0 & 0 & 0 & 0 & 0 & 1
\end{pmatrix}
\]
and
\[A_Q =  \begin{pmatrix}
1 & 1 & 1 & 1 & 1 & 1 & 1 & 1 \\ 
0 & 1 & 0 & 1 & 1 & 0 & 1 & 1 \\ 
0 & 0 & 1 & 0 & 0 & 0 & 0 & 1 \\ 
0 & 0 & 0 & 1 & 1 & 0 & 0 & 1 \\
0 & 0 & 0 & 0 & 1 & 0 & 0 & 0 \\ 
0 & 0 & 0 & 0 & 0 & 1 & 0 & 0 \\ 
0 & 0 & 0 & 0 & 0 & 0 & 1 & 0 \\
0 & 0 & 0 & 0 & 0 & 0 & 0 & 1
\end{pmatrix}
\]
respectively.

\begin{figure}[htb]
\centering
\includegraphics[scale=0.2]{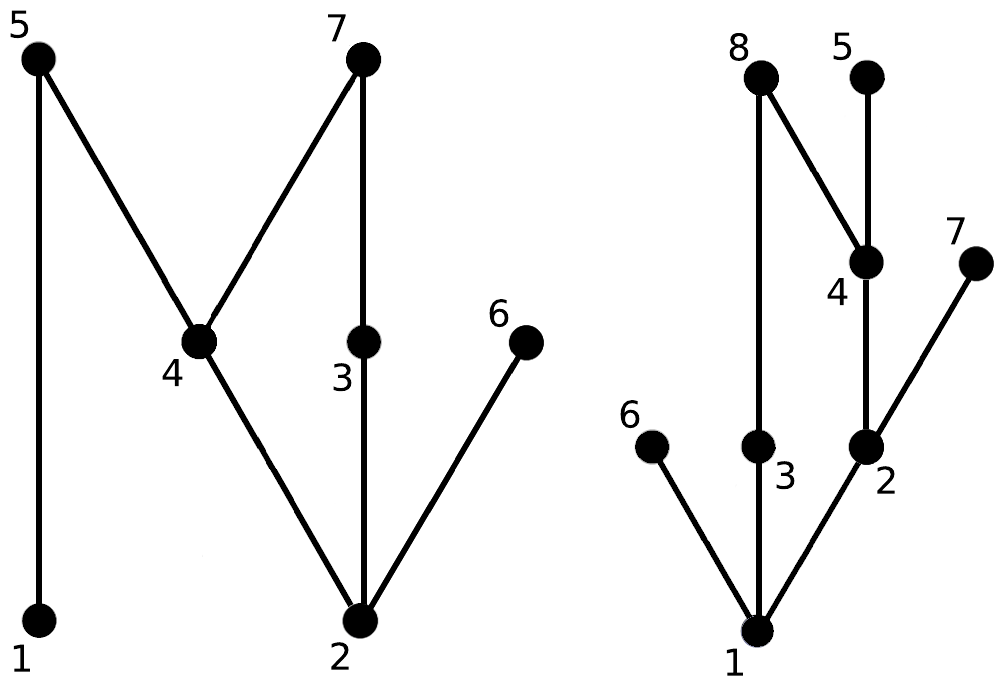}
\caption{Posets $P$ and $Q$ respectively.}
\end{figure}

The maximal elements of $Q$ can by found by searching the rows of $A_Q$ in which only one $1$ appears. In this case, the fifth, sixth, seventh, and eighth row. Thus, the radius matrix of $Q$ is
\[ \begin{pmatrix}
1 & 1 & 1 & 1 \\ 
1 & 0 & 1 & 1 \\ 
0 & 0 & 0 & 1 \\ 
1 & 0 & 0 & 1 \\
1 & 0 & 0 & 0 \\ 
0 & 1 & 0 & 0 \\ 
0 & 0 & 1 & 0 \\
0 & 0 & 0 & 1
\end{pmatrix}
.\]

Analogously, the radius matrix of $P$ is
\[ \begin{pmatrix}
1 & 0 & 0  \\ 
1 & 1 & 1  \\ 
0 & 0 & 1  \\ 
1 & 0 & 1  \\
1 & 0 & 0  \\ 
0 & 1 & 0  \\ 
0 & 0 & 1  
\end{pmatrix}
.\]

Adding a row of zeros between the fifth and sixth row, and then, adding a column of zeros between the first and second column, we obtain the following ER-matrix of $P$:
\[ \begin{pmatrix}
1 & 0 & 0 & 0 \\ 
1 & 0 & 1 & 1 \\ 
0 & 0 & 0 & 1 \\ 
1 & 0 & 0 & 1 \\
1 & 0 & 0 & 0 \\ 
0 & 0 & 0 & 0 \\ 
0 & 0 & 1 & 0 \\
0 & 0 & 0 & 1
\end{pmatrix}
.\]
But this matrix's ideal is contained in the ideal of the radius matrix of $Q$. Thus, $R(P) \leq R(Q)$.
\end{myexe}

%\bibliographystyle{IEEEtran}

%\bibliography{IEEEabrv,ref}


\end{document}